\documentclass{amsart}

\usepackage[normalem]{ulem}

\usepackage{latexsym} 
\usepackage{amsmath} 
\usepackage{epsfig}
\usepackage{amssymb}
\usepackage{enumerate}
\usepackage{color}
\usepackage{graphicx}
\usepackage{hyperref}

\usepackage[dvipsnames]{xcolor}

\newtheorem{theorem}{Theorem}[section]

\newtheorem{lemma}[theorem]{Lemma}
\newtheorem{sublemma}{}[theorem]
\newtheorem{corollary}[theorem]{Corollary}

\newtheorem{observation}[theorem]{Observation}
\newtheoremstyle{problem}{}{}{}{0pt}{}{}{0pt}{}
\theoremstyle{problem}

\newcommand{\cC}{{\mathcal C}}
\newcommand{\cF}{{\mathcal F}}
\newcommand{\cG}{{\mathcal G}}

\newcommand{\cI}{{\mathcal I}}

\newcommand{\cN}{{\mathcal N}}
\newcommand{\cO}{{\mathcal O}}

\newcommand{\cU}{{\mathcal U}}

\title[]{On the existence of funneled orientations for classes of rooted phylogenetic networks}

\author{Janosch D\"ocker and Simone Linz} 

\thanks{We thank the New Zealand Marsden Fund for their financial support.}

\address{School of Computer Science, University of Auckland, Auckland, New Zealand}
\email{janosch.doecker@auckland.ac.nz}

\address{School of Computer Science, University of Auckland, Auckland, New Zealand}
\email{s.linz@auckland.ac.nz}


\date{\today}

\begin{document}

\begin{abstract}
Recently, there has been a growing interest in the relationships between unrooted  and rooted phylogenetic networks. In this context, a natural question to ask is if an unrooted phylogenetic network $\cU$ can be oriented as a rooted phylogenetic network such that the latter satisfies certain structural properties. In a recent preprint, Bulteau et al. claim that it is computational hard to decide if $\cU$ has a funneled (resp.~funneled tree-child) orientation, for when the internal vertices of $\cU$ have degree at most 5. Unfortunately, the proof of  their funneled tree-child result appears to be incorrect. In this paper, we present a corrected proof and show that hardness remains for other popular classes of rooted phylogenetic networks such as funneled normal and funneled reticulation-visible. Additionally, our results hold regardless of whether $\cU$ is rooted at an existing vertex or by subdividing an edge with the root.
\end{abstract}

\maketitle

\section{Introduction}

Phylogenetic networks are commonly used to represent evolutionary relationships between taxa such as species, individuals of a population, or viruses. In general terms, phylogenetic networks are graphs whose vertices represent taxa and edges represent inferred evolutionary relationships. Unrooted phylogenetic networks are undirected graphs that do not contain any explicit information about the direction of evolution such as ancestor-descendant relationships. To include such information, rooted phylogenetic networks are used which are directed acyclic graphs with a single source (the root) and, as in the unrooted case, leaves representing extant taxa, and internal vertices representing hypothetical or extinct taxa. In addition, there also exists a wide range of well-studied classes of phylogenetic networks that are each characterised by certain structural properties such as level-$k$ and tree-based networks, which have been defined for rooted and unrooted phylogenetic networks, as well as  tree-child and  reticulation-visible networks, which have been defined for rooted phylogenetic networks only. For a comprehensive overview of phylogenetic network classes, we refer the interested reader to~\cite{kong22}. Network classes that are used throughout this paper are formally defined in the next section.

Initiated by a paper by Huber et al.~\cite{huber22}, there has recently been a growing interest in the relationships between unrooted and rooted phylogenetic networks. More specifically, the authors  have investigated several questions in the context of orienting an unrooted phylogenetic network $\cU$ as a rooted phylogenetic network. The process of orienting $\cU$ consists of subdividing an edge of $\cU$ with a new vertex $\rho$ and assigning a direction to each edge such that the resulting directed graph is a rooted phylogenetic network with root $\rho$. Given $\cU$, results for the following decision problems have been established in~\cite{huber22}: 
\begin{enumerate}[(i)]
\item \textsc{Constrained Orientation}: If the desired in-degree for each vertex of $\cU$ is given as well as an edge of $\cU$ to be subdivided with $\rho$, can $\cU$ be oriented as a rooted phylogenetic network with root $\rho$ that satisfies all in-degree constraints?   
\item \textsc{$\cC$-Orientation}: Can $\cU$ be oriented such that the resulting digraph belongs to a given class~$\cC$ of phylogenetic networks?  
\end{enumerate}
More specifically, Huber et al.~\cite{huber22} present a polynomial-time algorithm for \textsc{Constrained Orientation} that computes an orientation of $\cU$ satisfying the given in-degree constraints or outputs that there is none.~Furthermore, they have shown that, if $\cU$ is binary  (i.e., each internal vertex has degree 3), then it is NP-complete to decide if $\cU$ has a tree-based orientation and, lastly, if $\cU$ is binary  and $\cC$ satisfies certain properties, then \textsc{$\cC$-Orientation} is fixed-parameter tractable with respect to the so-called level of $\cU$. Notably, the latter result includes the popular class of binary tree-child networks. In the same paper, the authors also state two open questions.
\begin{enumerate}[Q1]
\item Given an unrooted binary phylogenetic network $\cU$, can it be decided in polynomial time if $\cU$ has a tree-child orientation?

\item Given an unrooted phylogenetic network $\cU$, can it be decided in polynomial time if $\cU$ has a funneled orientation?
\end{enumerate}

\noindent The term {\it funneled} was coined by Huber et al.~\cite[p.\ 26]{huber22} and refers to the restriction that each reticulation (i.e., a vertex of a rooted phylogenetic network with in-degree at least 2) has out-degree 1. As noted by the authors, it is common in the literature to define reticulations in this way.

Recently Bulteau et al.~\cite[Cor.\ 4]{bulteau23} have shown that it is NP-complete to decide if $\cU$ has a funneled orientation if each vertex of $\cU$ has degree at most 5 and a vertex of $\cU$ is chosen as root (instead of subdividing an edge of $\cU$ with the root). 
 We remark that Bulteau et al. refer to a funneled orientation as a {\it valid} orientation and that their construction to establish NP-completeness allows for vertices of degree 2 in $\cU$ which are explicitly excluded in the definition of such a network in~\cite{huber22} as well as in most other literature on phylogenetic networks because degree-2 vertices do not carry any biological meaning. Relatedly, Garvardt et al.~\cite{garvardt23} have analysed the parameterised complexity of a variant of the {\sc Constraint Orientation} problem, where the in-degree of each vertex $v$ is not fixed but can take on a value of a list that is associated with $v$. On the positive side, Bulteau et al.~\cite[Cor.\ 3]{bulteau23}  have shown that Q2 can be answered affirmatively if $\cU$ has degree at most 3. In particular, they have established a linear-time algorithm that computes a funneled orientation of $\cU$ if it exists or, otherwise, that correctly concludes that no such  orientation exists.
This last result is in line with an earlier result by Janssen et al.~\cite[Lem.\ 4.13 and 4.14]{janssen18} who characterised those unrooted binary phylogenetic networks that can be oriented as rooted binary phylogenetic networks. Inspecting the proof of their characterisation~\cite[Lem.\ 4.13]{janssen18}, their result also yields a polynomial-time algorithm for deciding if an unrooted binary phylogenetic network can be oriented as a rooted binary phylogenetic network.

Turning to Q1, Maeda et al.~\cite{maeda23} have recently presented several necessary conditions for when an unrooted binary phylogenetic network has a tree-child orientation. It is worth noting that, in the binary case, any orientation as a rooted phylogenetic network is also funneled. 
Moreover, given an unrooted phylogenetic network $\cU$ whose non-leaf vertices have degree at least 2 and at most 5 and given a designated root vertex $v$ of $\cU$, Bulteau et al.~\cite[Cor.\ 5]{bulteau23} claim that it is NP-hard to decide if $\cU$ has a funneled tree-child orientation with root $v$. Unfortunately, as we will show later, the proof of this corollary is incorrect. In summary, Q1 remains open.

In this paper, we make a step towards answering Q1 by presenting a corrected proof that establishes NP-completeness for deciding if an unrooted phylogenetic network with maximum degree 5 has a funneled tree-child orientation.~In comparison to Bulteau et al.~\cite{bulteau23}, our result does not require a designated root to be part of the input and our construction yields unrooted phylogenetic networks without degree-2 vertices. Although our result builds on Bulteau et al., correcting their proof requires novel gadgets and careful reasoning. We also show that our result extends to several other classes of funneled orientations such as funneled tree-sibling and funneled normal networks.

The remainder of the paper is organised as follows. Section~\ref{sec:preliminaries} contains definitions and formal problem statements. In Section~\ref{sec:tree-child-orientations}, we show that, although the proof by Bulteau et al.~\cite[Cor.\ 5]{bulteau23} is  incorrect, the result that it is NP-complete to decide if an unrooted phylogenetic network $\cU$ with degree at most 5 has a funneled tree-child orientation can be recovered. We then show  in Section~\ref{sec:orientation-other-classes} that hardness remains if each non-leaf vertex of $\cU$ has degree exactly 5. This latter result is subsequently used to  establish NP-completeness for deciding if $\cU$ has a funneled tree-sibling, funneled reticulation-visible, or funneled normal orientation. We remark that all results hold under two rooting variants: Variant $A$ as introduced by Huber et al.~\cite{huber22} subdivides an edge of $\cU$ with the root, and Variant $B$ chooses a vertex of $\cU$ to be the root. If $\cU$ is binary, one typically chooses a leaf of $\cU$ to be the root since, otherwise, the resulting orientation is not binary. Indeed, all constructions presented in this paper choose a leaf of $\cU$ as the root when establishing hardness under Variant $B$.

\section{Preliminaries}\label{sec:preliminaries}
In this section, we introduce notation and terminology that is used throughout the rest of the paper, and formally state the decision problems whose computational complexity we investigate in subsequent sections.

\subsection{Phylogenetic networks.} Let $X$ be a non-empty finite set. An \emph{unrooted phylogenetic network} $\cU$ on $X$ is a simple undirected graph with no degree-2 vertex and with a bijection between the vertices in $\cU$ that have degree 1 and the set $X$, that is, the \emph{leaves} of $\cU$ are bijectively labelled with elements from $X$. We will freely refer to leaves using their labels. Two distinct vertices $v$ and $w$ of $\cU$ are called \emph{neighbours} if $\{v, w\}$ is an edge in $\cU$. 
 
A \emph{rooted phylogenetic network} $\cN$ on $X$ is a directed acyclic graph with no loop and no parallel arcs that satisfies the following properties:
\begin{enumerate}[(i)]
\item there is a unique vertex $\rho$, the \emph{root}, with in-degree 0 and out-degree at least~$1$,
\item a vertex of out-degree 0 has in-degree 1 and the set of vertices with out-degree~0 is $X$, and
\item each internal vertex has either in-degree 1 and out-degree at least 2, or in-degree at least 2 and out-degree at least 1.
\end{enumerate}

\noindent For technical reasons, we sometimes consider a class of networks that is more general than the class of rooted (resp. unrooted) phylogenetic networks. Specifically, we refer to a network that is a subdivision of a rooted phylogenetic network on $X$ (resp. unrooted phylogenetic network on $X$) as a {\it rooted pseudo network} on $X$ (resp. {\it unrooted pseudo network} on $X$).

Now, let $\cN$ be a rooted pseudo network on $X$. For an arc $e = (u,v)$ in $\cN$, we say that $u$ is a \emph{parent} of $v$ and $v$ is a \emph{child} of $u$. If two distinct vertices $v$ and $w$ of $\cN$ have a common parent, we say that $v$ and $w$ are \emph{siblings}. Moreover, a vertex of $\cN$ is called a \emph{tree vertex} if it has in-degree $1$ and out-degree at least $1$ and is called a \emph{reticulation} if it has in-degree at least $2$ and out-degree at least $1$. Lastly, an arc $(u,v)$ is called a \emph{shortcut} in $\cN$ if  there is a directed path from $u$ to $v$ in $\cN$ that does not contain $(u,v)$.     

We next define five classes of rooted pseudo networks. These classes are typically defined in the context of rooted phylogenetic networks. However, their definitions naturally carry over to rooted pseudo networks as follows. Let $\cN$ be a rooted pseudo network. We say that $\cN$ is
\begin{enumerate}[(i)]
\item \emph{tree-child} if each non-leaf vertex of $\cN$ has a child that is a tree vertex or a leaf,
\item {\it normal} if $\cN$ is tree-child and does not contain a shortcut,
\item {\it tree-sibling} if each reticulation of $\cN$ has a sibling that is a tree vertex or a leaf,
\item \emph{reticulation-visible} if, for each reticulation $v$ of $\cN$, there is a leaf $\ell$ such that each directed path from the root of $\cN$ to $\ell$ contains $v$, and
\item \emph{funneled} if each reticulation of $\cN$ has out-degree 1. 
\end{enumerate}
Since the above five classes are only defined for rooted pseudo networks, we will omit the adjective `rooted' when referring to them. For more details on these network classes, see~\cite{cardona09, cardona08, iersel10,kong22,willson10}.

\subsection {Orientations.}  \label{sec:orientations}
A \emph{connector network} $\cG_k$ with $k\geq 0$ is a graph that can be obtained from an unrooted pseudo network $\cU$ by deleting the  label of $k$ leaves in~$\cU$. We call a degree-1 vertex of $\cG_k$  a {\it connector leaf} or {\it unlabelled leaf} if it has no label, and a {\it non-connector leaf} or {\it labelled leaf} otherwise. Furthermore, we refer to $\cU$ as the \emph{partner network} of $\cG_k$. Note that $\cU$ is the unique partner network of $\cG_k$ up to those leaf labels that exist in $\cU$ and not in $\cG_k$. Lastly, for $k=1$, let $\cG_1$ be a connector network. We say that an unrooted pseudo network $\cU$ contains $\cG_1$ as a \emph{pending subgraph} if $\cU$ can be obtained from $\cG_1$ by identifying its connector leaf with a vertex $v$ of some unrooted pseudo network $\cU'$ and deleting the label if $v$ is a leaf of $\cU'$. For example, the unrooted pseudo network that underlies the network shown in Figure~\ref{fig:root-forcing}(ii) has both $\cG_1$ and $\cG_1'$ that are shown in (i) of the same figure as a pending subgraph. 

For the purpose of the upcoming definitions, let $\cG_k$ be a connector network with $k\geq 0$. Furthermore,  let $L$ be the set of all labelled leaves of $\cG_k$, and let $U$ be the set of all unlabelled leaves of $\cG_k$. We next define a process that assigns a direction to each edge of $\cG_k$ and, if $k\leq 1$, also introduces a root vertex $\rho$. 

First, for $k\leq 1$, an \emph{orientation} of $\cG_k$ is obtained by either 

\noindent {\bf Variant $A$.} subdividing an edge of $\cG_k$ with a new root vertex $\rho$ and then assigning a direction to each edge such that $\rho$ has in-degree 0 and each element in $L$ has out-degree 0, or

\noindent {\bf Variant $B$.} choosing a vertex $u$  of $\cG_k$  to be the root $\rho$ by setting $u=\rho$, deleting the label of $u$ if $u$ is a labelled leaf, and then assigning a direction to each edge such that $\rho$ has in-degree 0 and each element in $L$ if $u \notin L$  (resp. $L\backslash\{u\}$ if $u\in L$) has out-degree 0. 

\noindent Note that such an orientation always exists.
Now, let $\cC$ be a class of rooted (pseudo) networks and let $R\in\{A,B\}$. For $k=0$, we say that an unrooted pseudo network $\cG_0$ has a {\it $\cC_R$-orientation} or, equivalently, that $\cG_0$ is {\it $\cC_R$-orientable} if there exists an orientation $\cO$ of $\cG_0$ such that the following properties are satisfied.
\begin{enumerate}[(i)]
\item  $\cO$ is obtained from $\cG_0$ by following Variant $R$,
\item $\cO$ is a network in $\cC$ with root $\rho$,
\item if $\rho\notin L$, then $\cO$ has $|L|$ leaves, and
\item if  $\rho\in L$, then $\cO$ has $|L|-1$ leaves.
\end{enumerate}
Turning to $k=1$, we say that a connector network $\cG_1$  with a unique connector leaf $r$ has a {\it $\cC_R$-orientation} or, equivalently, that $\cG_1$ is {\it $\cC_R$-orientable} if its partner network has a $\cC_R$-orientation. For $R=A$, $\cG_1$ is called \emph{$\cC_A$-root-forcing} if  there is no $\cC_A$-orientation of $\cG_1$ that subdivides the (unique) edge incident with $r$ with a new root vertex. Similarly, for $R=B$, $\cG_1$ is called \emph{$\cC_B$-root-forcing} if there is no $\cC_B$-orientation of $\cG_1$ that chooses $r$ to be the root. 
\begin{figure}
\includegraphics[width=.6\textwidth]{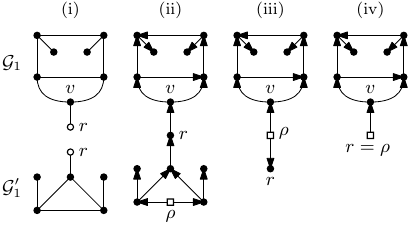}
\caption{(i) Two connector networks $\cG_1$ and $\cG'_1$ each with a single connector leaf $r$. 
(ii) A rooted pseudo network $\cN$  obtained by identifying $r$ in $\cG_1$ with $r$ in $\cG'_1$ and orienting the resulting unrooted pseudo network  according to Variant $A$. Observe that $\cN$ is not tree-child. (iii)  An orientation of $\cG_1$ following Variant $A$.  (iv)  An orientation of $\cG_1$ following Variant $B$.  As we will see in Lemma~\ref{lem:root_gadget}, for $\cC$ being the class of tree-child network, $\cG_1$ is $\cC_A$-root-forcing because every $\cC_A$-orientation of $\cG_1$ places the root on an edge that is not $\{r,v\}$.}
\label{fig:root-forcing}
\end{figure}  
To illustrate, see Figure~\ref{fig:root-forcing} for an example of a connector network with a single connector leaf that is $\cC_A$-root-forcing  for when $\cC$ is the class of tree-child networks. In Figure~\ref{fig:root-forcing} as well as in all other figures of this paper, the root is indicated by a small square, connector leaves are indicated by open circles whereas all other vertices are indicated by filled circles. Moreover, we frequently  label internal vertices in figures. Their only purpose is to make references. Indeed, they should not be regarded as genuine labels as those used for leaves of pseudo networks and non-connector leaves of connector networks. Furthermore, labels of leaves of pseudo networks and non-connector leaves of connector networks are sometimes omitted from figures to facilitate a clear presentation unless they are of importance (e.g. if they are explicitly mentioned in a proof).
 
Second, for $k\geq 2$, an {\it orientation} $\cO$ of $\cG_k$ is obtained by assigning a direction to each edge in $\cG_k$. 
Let $\cO$ be an acyclic orientation of $\cG_k$ such that each vertex in $L$ has out-degree 0, and each vertex that is not in $L\cup U$ has in-degree at least 1 and out-degree at least 1, then $\cO$  is called
\begin{enumerate}[(i)]
\item \emph{$\cF$-compatible} if each vertex of $\cO$ with in-degree at least 2 has out-degree 1,
\item \emph{$\mathcal{TC}$-compatible} if each vertex of $\cO$ that is not in $L\cup U$  has a child that is either in $L\cup U$ or a vertex with in-degree 1 and out-degree at least 1, and
\item  \emph{strongly $\mathcal{TC}$-compatible} if each vertex of $\cO$ that is not in $L\cup U$ has a child that is either in $L$ or a vertex with in-degree 1 and out-degree at least 1.  
\end{enumerate}

\subsection{Problem statements} Let $\cC$ be a class of rooted phylogenetic networks, and let $R\in\{A,B\}$. In this paper, we analyse the following decision problem. 

\begin{center}
\noindent\fbox{\parbox{.95\textwidth}{
\noindent\textsc{Funneled $\cC_R$-Orientation}\\
\textbf{Instance.} An unrooted phylogenetic network $\cU$.\\
\textbf{Question.} Does there exist a funneled $\cC_R$-orientation of $\cU$?
}}
\end{center}

\noindent If $R=B$, we emphasise that any vertex of $\cU$ can be chosen as root. In particular, no preselected root is part of the input to {\sc Funneled $\cC_B$-Orientation.} To obtain NP-hardness results for {\sc Funneled $\cC_R$-Orientation} in the next section, we design three gadgets, which are connector networks. 
Multiple copies of these gadgets are then combined such that the resulting graph is an unrooted pseudo network or an unrooted phylogenetic network $\cU$. We then orient each gadget separately such that, collectively, these orientations will result in a funneled $\cC_R$-orientation $\cN$ of $\cU$ if such an orientation exists. To provide some intuition for some of the more technical definitions given in Section~\ref{sec:orientations}, if a gadget is a connector network $\cG_1$ with exactly one connector leaf, then the upcoming hardness constructions enforce that the root of $\cN$ is contained in the subgraph of $\cN$ that is obtained by orienting $\cG_1$. On the other hand, if a gadget is a connector network $\cG_k$ with $k\geq 2$, then the hardness constructions enforce that the root of $\cN$ is not contained in the subgraph of $\cN$ that is obtained by orienting~$\cG_k$.

We now turn to the statement of two Boolean satisfiability problems that we later use to establish NP-hardness of \textsc{Funneled $\cC_R$-Orientation}. Let $V = \{x_1, x_2, \ldots, x_n\}$ be a set of variables. A \emph{literal} is a variable $x_i$ or its negation $\bar{x}_i$, and a \emph{clause} is a disjunction of a subset of $\{x_i, \bar{x}_i: i\in \{1, 2, \ldots, n\}\}$. If a clause contains exactly $k$ distinct literals for $k\ge 1$, then it is called a {\em $k$-clause}. We say that a clause is \emph{positive} if it is a subset of $V$. A {\it Boolean formula in conjunctive normal form} (or short, {\it Boolean formula}) is a conjunction of clauses, i.e., an expression of the form $\varphi = \bigwedge_{j = 1}^m c_j$, where $c_j$ is a clause for all $j$. Now, let $\varphi$ be a Boolean formula. We say that $\varphi$ is \emph{positive} if each clause of $\varphi$ is positive, i.e., no clause contains an element in $\{\bar{x}_1,\bar{x}_2,\ldots,\bar{x}_n\}$.  A {\em truth assignment} for $V$ is a mapping $\beta \colon V \rightarrow \{T, F\}$, where~$T$ represents the truth value True and $F$ represents the truth value False. A truth assignment $\beta$ \emph{satisfies} $\varphi$ if at least one literal of each clause evaluates to $T$ under $\beta$. If $\beta$  satisfies $\varphi$ and has the additional property that at least one literal of each clause evaluates to $F$, we say that $\beta$ \emph{nae-satisfies} $\varphi$. 

\begin{center}
\noindent\fbox{\parbox{.95\textwidth}{
\noindent\textsc{Positive Not-All-Equal $(2,3)$-SAT}\\
\textbf{Instance.} A set $V=\{x_1,x_2,\ldots,x_n\}$ of variables and a collection $C=\{c_1,c_2,\ldots,c_m\}$ of positive clauses over $V$ such that each clause consists of either two or three distinct variables.\\
\textbf{Question.} Is there a truth assignment for $V$ that nae-satisfies $C$? 
}}
\end{center}

\begin{center}
\noindent\fbox{\parbox{.95\textwidth}{
\noindent\textsc{Positive 1-in-3 SAT}\\
\textbf{Instance.} A set  $V=\{x_1,x_2,\ldots,x_n\}$ of variables and a collection $C=\{c_1,c_2,\ldots,c_m\}$ of positive clauses over $V$ such that each clause consists of 3 distinct variables.\\
\textbf{Question.} Is there a truth assignment for $V$ that sets exactly one variable in each clause in $C$ to $T$? 
}}
\end{center}

\section{Funneled Tree-Child Orientations}\label{sec:tree-child-orientations}
In this section, we analyse the aforementioned result by Bulteau et al.~\cite{bulteau23} and establish NP-completeness of {\sc Funneled $\cC_R$-Orientation} for when an unrooted phylogenetic network with maximum degree 5 is given and $\cC$ is the class of rooted phylogenetic networks that are tree-child.

\subsection{Construction by Bulteau et al.}\label{sec:bulteau}

As mentioned in the introduction, an incorrect proof has recently appeared in the literature~\cite{bulteau23} that claims NP-completeness for the following decision problem which is slightly different from  {\sc Funneled $\mathcal{C}_R$-Orientation}: Given an unrooted pseudo network $\cU$ on $X$ with maximum degree 5 and a vertex $v$ of $\cU$, does there exist an orientation $\cO$ of $\cU$ such that $\cO$ is a rooted pseudo network on $X$ (resp. $X\backslash\{v\}$ if $v\in X$) with root $v$ that is funneled and tree-child? 

We next describe the high-level idea of the construction that is presented in~\cite{bulteau23} and then explain why the tree-child orientation result does not hold. Given an instance of \textsc{Positive Not-All-Equal $(2,3)$-SAT}, the authors construct an unrooted pseudo network $\cU$ with leaf set $X=\{\ell_1,\ell_2,\ldots,\ell_n\}$. The construction of $\cU$ is based on the bipartite incidence graph which has an internal {\it variable vertex} $x_i$ for each variable, an internal {\it clause vertex $c_j$} for each clause, and an edge connecting $x_i$ with $c_j$ if and only if $x_i$ appears in $c_j$. Additionally, a leaf $\ell_i$ is attached to each $x_i$ and a vertex $r$ is introduced that is adjacent to each $x_i$. Lastly, to guarantee that $r$ has degree at most 5, additional edges and vertices are added. Importantly, the introduction of these additional edges and vertices does not change the edges that connect a variable vertex with a clause vertex. To keep the upcoming discussion as simple as possible, we omit this last step in the construction from the discussion and instead view $r$ as a high-degree vertex. As an example, consider the following yes-instance of \textsc{Positive Not-All-Equal $(2,3)$-SAT}, which is also detailed in~\cite{bulteau23}.
\begin{equation}\label{example:bulteau}
(x_1\vee  x_2\vee  x_3)\wedge (x_2\vee x_4)\wedge (x_1\vee x_4)\wedge (x_1\vee x_3)\wedge (x_2\vee x_3\vee x_4).
\end{equation}
For this instance, the construction results in the unrooted pseudo network that is shown in Figure~\ref{fig:counterexample}(a). 

\begin{figure}[t]
\includegraphics[width=\textwidth]{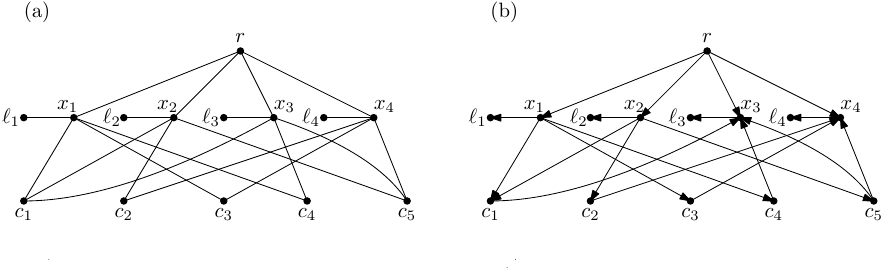}
\caption{(a) The construction of an unrooted pseudo network $\cU$ for the \textsc{Positive Not-All-Equal $(2,3)$-SAT} instance (1), and (b) an orientation $\cO$ of $\cU$ that is rooted at $r$.~Observe that $\cO$ is funneled but not tree-child.}
\label{fig:counterexample}
\end{figure} 

Now, let $\cU$ be an unrooted pseudo network obtained from the above construction for some instance $\cI$ of \textsc{Positive Not-All-Equal $(2,3)$-SAT}, and let $\cO$ be an orientation of $\cU$.  Bulteau at al.  essentially thought to have proved that  $\cO$ is a rooted pseudo network on $X$  with root $r$ that is funneled if and only if  $\cO$ is a rooted pseudo network on $X$  with root $r$ that is funneled and tree-child. Suppose that $\cO$ is a rooted pseudo network on $X$  with root $r$ that is funneled.  We next briefly describe some properties of $\cO$. Clearly, there are arcs $(r,x_i)$ and $(x_i,\ell_i)$ for each variable vertex $x_i$ with $i\in\{1,2,\ldots,n\}$. Since $\cO$ is funneled, the edges that join a variable vertex $x_i$ with a clause vertex $c_j$ are either all directed away from $x_i$ or all directed into $x_i$. Moreover, there is no clause vertex $c_j$ that has all arcs directed into it or all arcs directed away from it in $\cO$. Intuitively, $\cO$ can be translated into a truth assignment that nae-satisfies $C$, where arcs directed into a clause vertex correspond to true variables and arcs directed out of a clause vertex correspond to false variables. Hence, each clause contains at least one true variable and at least one false variable. Returning to the \textsc{Positive Not-All-Equal $(2,3)$-SAT} instance~(\ref{example:bulteau}) that is described above,  Figure~\ref{fig:counterexample}(b) shows an orientation that is a rooted pseudo network on four leaves and  with root $r$ that is funneled. However, this orientation is not tree-child because, for example, $c_2$ has no child that is a tree vertex or a leaf. In general, let $u$ be a child of $c_j$ in $\cO$. Note that $u=x_i$ for some variable vertex $x_i$. Since $x_i$ is also a child of $r$, it follows that the in-degree of $x_i$ is at least 2. Hence, $x_i$ is a reticulation and, because the same argument applies to any other child of $c_j$, it follows that $c_j$ has no child that is a tree vertex or a leaf. The erroneous statement in~\cite[p.\ 10]{bulteau23} is the following:
\begin{quote}
{\it ``Since each clause contains at least one true variable, each clause vertex has at least one child that is a tree node.''}
\end{quote}  
The correct conclusion would be that each clause vertex $c_j$ has at least one \emph{parent} that is a tree vertex.

\subsection{Funneled tree-child orientations for pseudo networks with maximum degree 5}

Let $\cC$ be the class of rooted pseudo networks that are tree-child. In this subsection, we establish hardness for the following decision problem. 

\begin{center}
\noindent\fbox{\parbox{.95\textwidth}{
\noindent\textsc{Funneled $\mathcal{C}_A$-Orientation ($\leq 5$, pseudo)}\\
\textbf{Instance.} An unrooted pseudo network $\cU$ with maximum degree 5.\\
\textbf{Question.} Does there exist a funneled $\mathcal{C}_A$-orientation of $\cU$?
}}
\end{center}

We start by introducing three gadgets that we call the {\it root gadget, connection gadget}, and {\it caterpillar gadget}. These gadgets are shown on the left-hand 
side of Figure~\ref{fig:root_gadget},~\ref{fig:connection_gadget}, and~\ref{fig:caterpillar_gadget}, respectively. They will play a central role in reducing {\sc Positive Not-All-Equal $(2,3)$-SAT} to {\sc Funneled $\mathcal{C}_A$-Orientation ($\leq 5$, pseudo)}, thereby establishing NP-completeness for the latter. 

The next three lemmas establish  properties of the root, connection, and caterpillar gadget. Throughout the proofs, we use the vertex labels as shown in Figures~\ref{fig:root_gadget},~\ref{fig:connection_gadget}, and~\ref{fig:caterpillar_gadget}, when referring to vertices of these gadgets.

\begin{figure}[t]
\includegraphics[width=.6\textwidth]{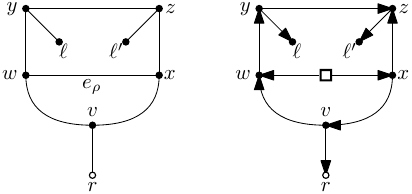}
\caption{Left: Root gadget with a single connector leaf $r$, and two non-connector leaves $\ell$ and $\ell'$. Right: An orientation of the root gadget.}
\label{fig:root_gadget}
\end{figure}

\begin{lemma}\label{lem:root_gadget}
Let $\cG$ be the root gadget, and let $\cC$ be the class of rooted pseudo networks that are tree-child. Then 
\begin{enumerate}[(i)]
\item $\cG$ has a $\mathcal{C}_A$-orientation whose root $\rho$ subdivides $\{w,x\}$, and
\item $\cG$  is $\cC_A$-root-forcing. 

\end{enumerate}
\end{lemma}
\begin{proof}
The orientation that is shown on the right-hand side of Figure~\ref{fig:root_gadget} is a $\mathcal{C}_A$-orientation of the partner network of $\cG$ which only differs from $\cG$ by viewing $r$ as a non-connector leaf. Hence, $\cG$ has a $\mathcal{C}_A$-orientation whose root $\rho$ subdivides $\{w,x\}$. This establishes  (i). Now, assume towards a contradiction that $\cG$ does not satisfy (ii). Then there exists a $\cC_A$-orientation $\cO$ of $\cG$ whose root $\rho$ subdivides $\{r,v\}$. Then $(\rho, v)$ is an arc in $\cO$. Furthermore, as $\ell$ and $\ell'$ are non-connector leaves, $(y,\ell)$ and $(z,\ell')$ are arcs in $\cO$. We next distinguish two cases.

First, suppose that $v$ has in-degree 1 and out-degree 2 in $\cO$. Then, $\rho$ is the parent of $v$, and $w$ and $x$ are the children of $v$. Observe that either $(w,x)$ or $(x,w)$ is an arc in $\cO$. By symmetry, we may assume that $(x,w)$ is an arc in $\cO$.  This implies that $w$ has in-degree 2 and out-degree 1. Moreover, since $\cO$ is a $\cC_A$-orientation of  $\cG$, it follows that each of $x$ and $y$ has in-degree 1 and out-degree 2. In turn $(y, z)$ and $(x,z)$ are arcs in $\cO$. Hence, $z$ has in-degree 2 and, thus, both children $w$ and $z$ of $x$ have in-degree 2; thereby contradicting that  $\cO$ is a $\cC_A$-orientation of  $\cG$. Second, suppose that $v$ has in-degree 2 and out-degree 1 in  $\cO$. By symmetry, we may assume that $x$ is the second parent of $v$ and $w$ is the child of $v$. Again, as $\cO$ is a $\cC_A$-orientation of  $\cG$, $w$ has in-degree 1 and out-degree 2. This implies that $(w, x)$ is an arc in $\cO$ and, so, $(x,v), (v, w),$ and $(w,x)$ are the arcs of a directed cycle; thereby contradicting that $\cO$ is acyclic. Combining both cases establishes the lemma.
\end{proof}

\begin{figure} 
\includegraphics[width=.7\textwidth]{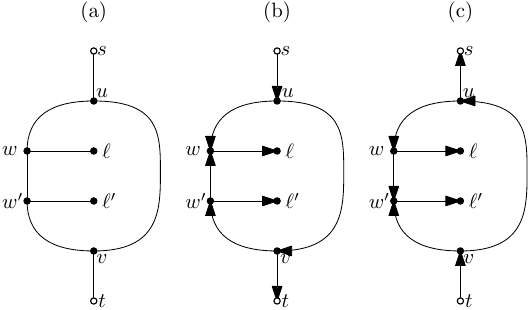}
\caption{Connection gadget (a) with two connector leaves $s$ and $t$, and two non-connector leaves $\ell$ and $\ell'$. Two orientations (b) and (c) of the connection gadget. Both orientations are $\cF$-compatible and strongly $\mathcal{TC}$-compatible.}
\label{fig:connection_gadget}
\end{figure}

The next corollary is a generalisation of Lemma~\ref{lem:root_gadget}(ii).

\begin{corollary}\label{c:root_gadget}
Let $\cC$ be the class of rooted pseudo networks that are tree-child. Let $\cU$ be an unrooted pseudo network that contains a pending subgraph isomorphic to the root gadget (up to the labelling of the two non-connector leaves of that gadget). If $\cU$ has a funneled $\cC_A$-orientation, then the root subdivides an edge of the root gadget that is not $\{r,v\}$.
\end{corollary}

\begin{proof}
Observe that the connector leaf $r$  is either an internal vertex or a leaf of $\cU$. Assume that there exists a funneled $\cC_A$-orientation $\cO$ of $\cU$ whose root $\rho$ subdivides $\{r,v\}$ or an edge that is not an edge of the root gadget. Since there exists a directed path from $\rho$ to leaf $\ell$ of the root gadget, it follows that $(r,v)$ is an arc of $\cO$. Now the proof can be established as the proof of Lemma~\ref{lem:root_gadget}(ii).
\end{proof}

\begin{lemma}\label{lem:connection_gadget}
Let $\cG$ be the connection gadget. Then $\cG$ satisfies the following properties:
\begin{enumerate}[(i)]
\item There exists an orientation of $\cG$ that is $\cF$-compatible and strongly $\mathcal{TC}$-compatible in which each of $u$ and $v$ has in-degree 1 and out-degree 2, and $(s, u)$ and $(v, t)$ are arcs.  
\item There exists an orientation of $\cG$ that is $\cF$-compatible and strongly $\mathcal{TC}$-compatible in which each of $u$ and $v$ has in-degree 1 and out-degree 2, and $(u, s)$ and $(t, v)$ are arcs. 
\item Each $\mathcal{TC}$-compatible orientation of $\cG$ has either arcs $(s, u)$ and $(v, t)$ or arcs $(u, s)$ and $(t, v)$.   
\end{enumerate}
\end{lemma}
\begin{proof}
Since $\cG$ has two connector leaves, recall that, by definition, an orientation of $\cG$ is obtained by assigning a direction to each edge of $\cG$. The orientations of $\cG$ that are shown in Figure~\ref{fig:connection_gadget}(b) and (c) establish (i) and (ii), respectively. Next we show that $\cG$ also satisfies (iii). Let $\cO$ be a $\mathcal{TC}$-compatible orientation of $\cG$. Due to (i), such an orientation exists. Furthermore, by definition of $\mathcal{TC}$-compatibility, $(w,\ell)$ and $(w',\ell')$ are arcs in $\cO$.

First, suppose that $(u, s)$ and $(v, t)$ are arcs in $\cO$. By symmetry, we may assume that $(u,v)$ is an arc in $\cO$. In turn, this implies the arc $(w,u)$ exists because, otherwise, $u$ has in-degree 0. Using the same argument again, $(w',w)$ and $(v,w')$ are also  arcs in $\cO$. Now $(u,v),(v,w'),(w',w)$ and $(w,u)$ are the arcs of a directed cycle in $\cO$; thereby contradicting that $\cO$ is acyclic.  Second, suppose that $(s, u)$ and $(t, v)$ are arcs in $\cO$. By symmetry, we may again assume that $(u,v)$ is an arc in $\cO$. Then, $(v,w')$ is an arc in $\cO$ because, otherwise $v$ has out-degree 0. Since $\cO$ is $\mathcal{TC}$-compatible and $v$ has in-degree 2 and out-degree 1, $w'$ has in-degree 1 and out-degree 2. Hence $(w',w)$ is an arc in $\cO$. Lastly, $(w,u)$ is an arc in $\cO$ because, otherwise, $u$ has no child that is in $L\cup U$ or has in-degree 1 and out-degree at least 1. Now $(u,v),(v,w'),(w',w)$ and $(w,u)$ are again the arcs of a directed cycle in $\cO$; another contradiction. By combining both cases, we deduce that, each $\mathcal{TC}$-compatible orientation of $\cG$ has either arcs  $(s, u)$ and $(v, t)$, or arcs $(u, s)$ and $(t, v)$. This establishes (iii) and, therefore, the lemma.  
\end{proof}

While the root and connection gadgets are new to this paper, the caterpillar gadget is based on what Bulteau et al.~\cite{bulteau23} refer to as {\it Rule 1}. This rule is a graph operation that reduces the degree of an unrooted phylogenetic pseudo network $\cU$ while preserving the existence of certain orientations of $\cU$. Essentially, Rule 1 is based on the observation that a vertex $v$ with two adjacent leaves and an arc that is directed into $v$ forces directions on all other edges incident with $v$ in every funneled orientation. More precisely, we have the following:

\begin{observation}\label{obs:rule1}
Let $\cO$ be an $\cF$-compatible orientation of a connector network $\cG_k$ with $k\geq 0$. Furthermore, let $v$ be a vertex of $\cO$ that is incident with $d$ arcs such that one arc is directed into $v$ and two arcs are directed out of $v$. Then $v$ has in-degree 1 and out-degree $d-1$.
\end{observation}

Before establishing properties of the caterpillar gadget in the next lemma, we note that the non-connector leaf $\ell$ of this gadget is not necessary to show that the lemma holds. Its only purpose is to turn $p_1$ into a degree-5 vertex.

\begin{figure}
\includegraphics[width=\textwidth]{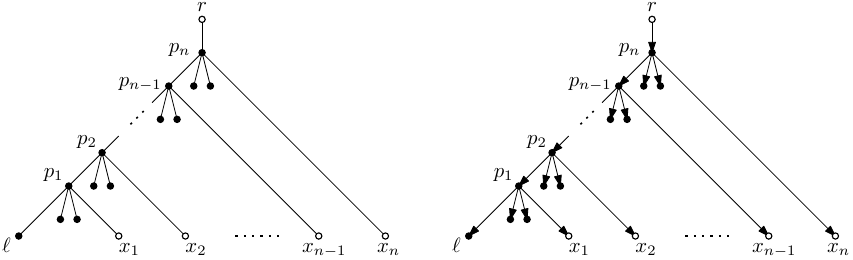}
\caption{Left: Caterpillar gadget with $n+1$ connector leaves, and $1+2n$ non-connector leaves. Right: An $\cF$-compatible and strongly $\mathcal{TC}$-compatible orientation of the caterpillar gadget.}
\label{fig:caterpillar_gadget}
\end{figure}

\begin{lemma}\label{lem:caterpillar_gadget}
For $n\geq 1$, let $\cG$ be the caterpillar gadget with degree-5 vertices $p_1,p_2,\ldots,p_n$. Then $\cG$ has a unique $\mathcal{F}$-compatible orientation $\mathcal{O}$ with arc $(r, p_n)$. Furthermore, \begin{enumerate}[(i)]
\item $\mathcal{O}$ is strongly $\mathcal{TC}$-compatible and 
\item for each $1 \leq i \leq n$, $(p_i, x_i)$ is an arc in $\cO$. 
\end{enumerate}
\end{lemma}
\begin{proof}
Suppose that  $(r, p_n)$ is an arc in an $\mathcal{F}$-compatible orientation $\cO$ of $\cG$. We first show that $\cO$ is unique. For each $i\in\{1,2,\ldots,n\}$, let $\ell_i$ and $\ell_i'$ be the two non-connector leaves adjacent to $p_i$ (these leaf labels are omitted in Figure~\ref{fig:caterpillar_gadget}). Since  $(r, p_n)$, $(p_n,\ell_n)$, and $(p_n,\ell_n')$ are arcs in $\cO$ and $\cO$ is $\mathcal{F}$-compatible, it follows from Observation~\ref{obs:rule1} that $(p_n,x_n)$ and $(p_n,p_{n-1})$ are arcs in $\cO$. Now, since $(p_n,p_{n-1})$, $(p_{n-1},\ell_{n-1})$, and $(p_{n-1},\ell_{n-1}')$ are arcs in $\cO$, it again follows by Observation~\ref{obs:rule1} that $(p_{n-1},x_{n-1})$ and $(p_{n-1},p_{n-2})$ are also arcs in $\cO$. Continuing in this way and noting that $p_1$ has, in addition to $\ell_1$ and $\ell_1'$, a third non-connector leaf $\ell$, it is easily checked that there  exists exactly one $\mathcal{F}$-compatible orientation for $\cG$ with arc $(r, p_n)$. Moreover, this orientation, which is illustrated on the right-hand side of Figure~\ref{fig:caterpillar_gadget},  also satisfies (ii) and, as no vertex in $\cO$ has in-degree at least 2, $\cO$ satisfies (i).
\end{proof}

The proof of the following theorem uses a construction that is based on ideas presented in Bulteau et al.~\cite{bulteau23} and summarised in Section~\ref{sec:bulteau}. However, we use a different root gadget and introduce a connection gadget that enables the tree-child property for orientations corresponding to nae-satisfying truth assignments.

\begin{theorem}\label{thm:hardness_at_most_5}

Let $\cC$ be the class of  rooted pseudo networks that are tree-child. Then \textsc{Funneled $\mathcal{C}_A$-Orientation ($\leq 5$, pseudo)} is NP-complete. 
\end{theorem}

\begin{proof}
Let $\cU$ be an unrooted pseudo network, and let $\cO$ be an orientation of $\cU$ following Variant $A$. Since it can be checked in polynomial time if $\cO$ is a $\mathcal{C}_A$-orientation of $\cU$, it follows that \textsc{Funneled $\mathcal{C}_A$-Orientation ($\leq 5$, pseudo)} is in NP. 
We next establish NP-hardness using a polynomial-time reduction from the variant of \textsc{Positive Not-All-Equal (2,3)-SAT}, where each variable appears exactly three times (see Dehghan et al.~\cite{dehghan15} for a proof that the problem remains NP-complete under this restriction). 

Let $\mathcal{I} = (V, C)$ be an instance of \textsc{Positive Not-All-Equal (2,3)-SAT} such that each variable appears in exactly three clauses. Let $m_2$ (resp. $m_3$) be the number of clauses in $C$ that contain exactly two (resp. three) distinct variables. Let $\cG_{n+1}^c$ be the caterpillar gadget with connector leaves $r,x_1, x_2,\ldots, x_n$,  let $\cG_1^r$ be the root gadget with connector leaf $r$, and let $\cG_2^1,\cG_2^2,\ldots,\cG_2^{2m_2+3m_3}$ be copies of the connection gadget each with connector leaves $s$ and $t$. We next construct an unrooted  pseudo network in the following way.

\begin{figure}
\includegraphics[width=.8\textwidth]{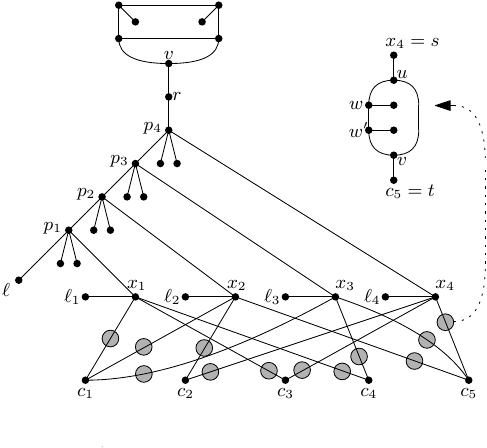}
\caption{The unrooted pseudo network that is constructed from the Boolean formula~(\ref{example:bulteau}) of \textsc{Positive Not-All-Equal (2,3)-SAT}. Details are given in the proof of Theorem~\ref{thm:hardness_at_most_5}. Connection gadgets are shown as grey shaded circles.}
\label{fig:construction}
\end{figure} 

\begin{enumerate}
\item Identify the connector leaf $r$ of $\cG_1^r$ with the connector leaf $r$ of $\cG_{n+1}^c$. We will continue to refer to the resulting degree-2 vertex as $r$. 
\item Attach a new leaf $\ell_i$ to each connector leaf $x_i$ with $i\in\{1,2,\ldots,n\}$ via a new edge. 
\item Let $c_1,c_2,\ldots,c_m$ be new vertices. Furthermore, let $i\in\{1,2,\ldots,n\}$, and let $j\in\{1,2,\ldots,m\}$. For each $x_i$ and $c_j$ such that $x_i$ is a variable of $c_j$, add one of the $2m_2+3m_3$ connection gadgets by identifying the connector leaf $s$ with $x_i$ and identifying the connector leaf $t$ with $c_j$.
\end{enumerate}

Let $\cU$ be the resulting unrooted pseudo network. We may assume for the remainder of the proof that all  leaves of $\cU$ have a distinct leaf label and that the leaf set of $\cU$ is $X$. Note that  $|X|=3+3n+2(2m_2+3m_3)$. By construction and since each variable appears exactly three times in $C$, the maximum degree of a vertex in $\cU$ is~5. Figure~\ref{fig:construction} shows an example of this construction for the Boolean formula given in~(\ref{example:bulteau}).

\begin{sublemma}\label{claim:hardness_at_most_5-sub}
$\cI$ is a yes-instance if and only if $\cU$ has a funneled $\mathcal{C}_A$-orientation.
\end{sublemma}
\begin{proof}
First, suppose that $\cI$ is a yes-instance. Let $\beta\colon V \rightarrow \{T, F\}$ be a truth assignment that nae-satisfies each clause in~$C$.
We construct a funneled $\cC_A$-orientation $\cO$ of $\cU$. Instead of referring to edges of $\cU$, we refer to edges of the root, connection, and caterpillar gadgets as shown in Figures~\ref{fig:root_gadget}--\ref{fig:caterpillar_gadget} when assigning directions because each edge in $\cU$, except for the edges incident with $\ell_i$ with $i\in\{1,2,\ldots,n\}$, corresponds to a unique edge in one of the gadgets that make up $\cU$. We start by subdividing the edge $\{w,x\}$ of $\cG_1^r$ with a new root vertex $\rho$ and direct the edges of the resulting root gadget as shown on the right-hand side of Figure~\ref{fig:root_gadget}. In particular, we have the arc $(v, r)$ in $\cO$. 
Turning to the edges of  $\cG_{n+1}^c$, we direct the edges of this gadget as shown on the right-hand side of Figure~\ref{fig:caterpillar_gadget}. In particular, we have the arc $(r,p_n)$ in $\cO$. Now, for each $i \in \{1,2,\ldots,n\}$, let $u_i, u_i'$, and $u_i''$ be the three neighbours of $x_i$ in $\cU$ that are vertices of three distinct connection gadgets. We direct the edge  $\{x_i,\ell_i\}$ into $\ell_i$ and, depending on $\beta(x_i)$, we direct  the three edges $\{x_i,u_i\}$, $\{x_i,u_i'\}$, and $\{x_i,u_i''\}$ in one of the following two ways:
\begin{enumerate}[(i)]
\item $(x_i, u_i), (x_i,u_i'), (x_i,u_i'')$ if $\beta(x_i) = T$, and
\item $(u_i, x_i), (u_i', x_i), (u_i'', x_i)$ if $\beta(x_i) = F$. 
\end{enumerate}
At this point, each connection gadget in $\cU$ has exactly one edge that is already assigned a direction and this edge is $e=\{s,u\}$ (where $u$ is the vertex as shown in Figure~\ref{fig:connection_gadget}). If $e$ is directed into $u$ in $\cO$, then direct the remaining edges of the connection gadget as shown in Figure~\ref{fig:connection_gadget}(b) and, otherwise, as shown in  Figure~\ref{fig:connection_gadget}(c). 
This completes the assignment of directions to the edges of $\cU$ and, therefore, the construction of $\cO$. Note that the set of vertices with out-degree 0 is $X$ and each such vertex has in-degree 1, and the unique vertex with in-degree 0 is $\rho$. We next show that $\cO$ is indeed a funneled $\cC_A$-orientation of $\cU$. In preparation for the upcoming arguments, let $$S=\{r, x_1,x_2,\ldots,x_n,c_1,c_2,\ldots,c_m\}.$$
Intuitively, $S$ contains each vertex in $\cO$ that is a connector leaf in one of the gadgets that are used in the construction of $\cU$. Furthermore, observe that each vertex $x_i$ in $\cO$ with $i\in\{1,2,\ldots,n\}$ has either in-degree 1 and out-degree 4, or in-degree 4 and out-degree 1.

We start by establishing that $\cO$ is acyclic. Since the orientations that are shown in Figures~\ref{fig:root_gadget}--\ref{fig:caterpillar_gadget} are acyclic, it is sufficient to show that no vertex in $S$ lies on a directed cycle. Let $v$ be the vertex of $\cG_1^r$ that is a neighbour of $r$ in $\cO$. Since there is no directed path from $p_n$ to $v$ in $\cO$, it follows that $r$ does not lie on a directed cycle.
Assume towards a contradiction that $x_i$ lies on a directed cycle for some $i\in\{1,2,\ldots,n\}$. Then, by the observation at the end of the last paragraph, $x_i$ has in-degree 1 and out-degree 4 and there is a directed path from $x_i$ to $p_i$. Hence, there is an arc $(x_{i'}, p_{i'})$ for some $i' \in \{1,2,\ldots,n\}$; a contradiction. Finally, if $c_j$ lies on a directed cycle for some $j\in\{1,2,\ldots,m\}$, then some $x_i$ lies on a directed cycle for some $i\in\{1,2,\ldots,n\}$; again a contradiction. We conclude that $\cO$ is acyclic and, thus, $\cO$ is a rooted pseudo network on $X$ with root $\rho$.

We next argue that $\cO$ is funneled, i.e., each vertex in $\cO$ with in-degree at least 2 has out-degree 1. It follows from the orientations shown in Figures~\ref{fig:root_gadget}--\ref{fig:caterpillar_gadget} that each vertex of $\cO$ that is not in $S$ satisfies this property. Turning to the vertices in $S$, $r$ has in-degree~1 and out-degree 1. Moreover, for each $i \in \{1,2,\ldots,n\}$, $x_i$ has in-degree 1 if $\beta(x_i) = T$ and out-degree 1 if $\beta(x_i) = F$. Finally, recall that each clause in $C$ has either two or three distinct variables. Since $\beta$  nae-satisfies $C$, it follows that $c_j$ has either in-degree 1 and out-degree is 2, or in-degree 2 and out-degree 1 for each $c_j\in C$. Hence, $\cO$ is funneled. 

We complete this part of the proof by showing that $\cO$ is tree-child. Since $r$ has in-degree 1 and out-degree 1 in $\cO$, it follows from Lemmas~\ref{lem:root_gadget}(i),~\ref{lem:connection_gadget}(i)--(ii), and \ref{lem:caterpillar_gadget}(i) that each vertex of $\cO$ that is not a leaf and not in $S$ has a child that is a leaf or a tree vertex. 
Furthermore, $r$ has a unique child $p_n$ that is a tree vertex, and each $x_i$ with $i\in\{1,2,\ldots,n\}$ is adjacent to a leaf $\ell_i$. Now consider a vertex $c_j$ in $\cO$ with $j\in\{1,2,\ldots,m\}$. Each child of $c_j$ is a vertex $v$ of a connection gadget. The edges of such a connection gadget are directed as shown in Figure~\ref{fig:connection_gadget}(c). Importantly, $v$ is a tree vertex.
We conclude that $\cO$ is a pseudo network on $X$ that is funneled and tree-child.

Second, let $\cO$ be a funneled $\cC_A$-orientation of $\cU$ with root $\rho$.  
By Corollary~\ref{c:root_gadget}, $\rho$ subdivides an edge of $\cG_1^r$ that is not $\{r,v\}$. Moreover, since $\cO$ is funneled and there is a directed path from $\rho$ to each leaf in $\cO$, $(v,r)$ and $(r,p_n)$ are arcs in $\cO$. By Lemma~\ref{lem:caterpillar_gadget}(ii), it now follows that $(p_i, x_i)$ is an arc in $\cO$ for each $i \in \{1,2,\ldots,n\}$. Hence, each $x_i$ has a parent $p_i$ and a child $\ell_i$ in $\cO$. We now define a truth assignment $\beta \colon V \rightarrow \{T, F\}$ as follows. For each $i \in \{1,2,\ldots,n\}$, we set $\beta(x_i) = T$ if $x_i$ is a tree vertex and $\beta(x_i) = F$ if $x_i$ is a reticulation in $\cO$.

Towards a contradiction, assume that $\beta$ does not nae-satisfy each clause in $C$. We distinguish two cases. First, suppose that there exists a clause $c_j=(x\vee x'\vee x'')$ or $c_j=(x\vee x')$ for some $j\in\{1,2,\ldots,m\}$ whose variables are all set to $T$ under $\beta$. We continue by assuming that $c_j=(x\vee x'\vee x'')$ and note that the same argument applies for when $c_j=(x\vee x')$. As $x, x'$ and $x''$ are tree vertices in $\cO$, each arc that joins a vertex in $\{x, x', x''\}$ with a vertex of a connection gadget is directed towards the latter. Now consider a directed path $P$ in $\cO$ from a vertex in $\{x, x', x''\}$ to $c_j$, and let $e$ be the last arc on $P$.
 Since $\cO$ is a $\cC_A$-orientation, the direction of the arcs restricted to a single connection gadget  in $\cO$ is $\mathcal{TC}$-compatible. Hence, Lemma~\ref{lem:connection_gadget}(iii) implies that $e$ is directed into $c_j$. It now follows that $c_j$ has in-degree 3 and out-degree 0; a contradiction. Second, suppose that there exists a clause  $c_{j'}\in C$ whose variables are all set to $F$ under $\beta$. Then, by an analogous argument, $c_{j'}$ has in-degree 0 and out-degree-3; also a contradiction. Thus, $\beta$ nae-satisfies each clause in $C$, that is, $\mathcal{I}$ is a yes-instance.   
\end{proof}

Noting that $\cU$ can be constructed in polynomial time and has a size that is polynomial in $n$ and $m$, we conclude the proof of the theorem. 
\end{proof}

\subsection{Hardness remains for phylogenetic networks}
Theorem~\ref{thm:hardness_at_most_5} shows that the erroneous result by Bulteau et al.~\cite[Cor.\ 5]{bulteau23} can be corrected by introducing additional gadgets. However, the construction results in unrooted and rooted pseudo networks. In this short subsection, we outline two modifications to the construction that avoid degree-2 vertices, thereby showing that the following decision problem is NP-complete for when $\cC$ is the class of tree-child networks. 

\begin{center}
\noindent\fbox{\parbox{.95\textwidth}{
\noindent\textsc{Funneled $\mathcal{C}_A$-Orientation ($3,5$)}\\
\textbf{Instance.} An unrooted phylogenetic network $\cU$ such that each non-leaf vertex has degree 3 or 5.\\
\textbf{Question.} Does there exist a funneled $\mathcal{C}_A$-orientation of $\cU$?
}}
\end{center}

Let $\cU$ be an unrooted pseudo network on $X$ that is constructed from an instance of  \textsc{Positive Not-All-Equal (2,3)-SAT}  as described in the proof of Theorem~\ref{thm:hardness_at_most_5}. By construction, each non-leaf vertex of $\cU$ has degree 2, 3, or 5. Moreover, the degree-2 vertices in $\cU$ are $r$ and $c_j$ for each $j\in\{1,2,\ldots,m\}$ such that $c_j$ is a clause with exactly two distinct variables. Consider a 2-clause $c_j = (x_i\vee x_{i'})$ with $i,i'\in\{1,2\ldots,n\}$ and $i\ne i'$. Observe that $c_j$ can be viewed as an XOR constraint on $x_i$ and $x_{i'}$, that is, exactly one of the two variables is set to $T$ in a truth assignment that nae-satisfies $c_j$. In $\cU$, this is simulated by the clause vertex $c_j$ being adjacent to two vertices that are part of two distinct connection gadgets, say $\cG_2^1$ and $\cG_2^2$. Without loss of generality, we assume that the vertex $s$ of $\cG_2^1$ has been identified with $x_i$ in the construction of $\cU$.  We refer to the operation of deleting all vertices of $\cG_2^2$ except for $c_j$ and $x_{i'}$, and then identifying $c_j$ with $x_{i'}$ as {\it the 2-clause suppression} of $c_j$. See Figure~\ref{fig:degree-2_gadget} for an example of a $2$-clause suppression. 

\begin{figure}
\includegraphics[width=.95\textwidth]{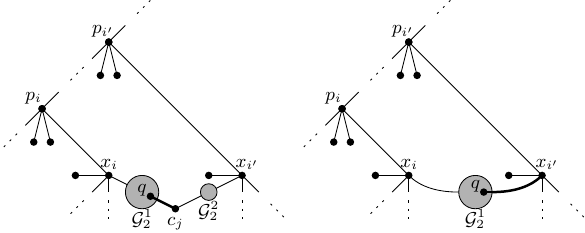}
\caption{Left: Partial construction of an unrooted pseudo network $\cU$ from an instance of \textsc{Positive Not-All-Equal (2,3)-SAT} that contains a clause $c_j=(x_i\vee x_{i'})$ as detailed in the proof of Theorem~\ref{thm:hardness_at_most_5}. Right: Unrooted pseudo network obtained from $\cU$ by the $2$-clause suppression of $c_j$. Connection gadgets are shown as grey shaded circles.}
\label{fig:degree-2_gadget}
\end{figure} 

\begin{lemma}\label{lem:help}
Let $\cC$ be the class of rooted pseudo networks that are tree-child. Let $\cU$ be an unrooted pseudo network that is constructed from an instance of  \textsc{Positive Not-All-Equal (2,3)-SAT} as described in the proof of Theorem~\ref{thm:hardness_at_most_5}. Let $\cU'$ be the unrooted phylogenetic network obtained from $\cU$ by suppressing $r$ and applying the 2-clause suppression to each clause vertex $c_j$ that has degree 2 in $\cU$. Then $\cU$ has a funneled $\cC_A$-orientation if and only if $\cU'$ has a funneled $\cC_A$-orientation.
\end{lemma}

\begin{proof}
We first note that if an unrooted phylogenetic network has a $\cC_A$-orientation, then this orientation is necessarily a tree-child network (without any vertex of in-degree 1 and out-degree 1).
Furthermore, by construction, observe that, except for $\{v, p_n\}$, each edge in $\cU'$ corresponds to a unique edge in $\cU$. In particular, for each $2$-clause $c_j=(x_i\vee x_{i'})$, the edge $\{q,x_{i'}\}$ in $\cU'$, where $q$ is a vertex of a connection gadget $\cG_2^1$ in $\cU$ corresponds to the edge $\{q,c_j\}$. These two edges are indicated by thick lines in Figure~\ref{fig:degree-2_gadget}.

Now, first, suppose that $\cO$ is a funneled $\cC_A$-orientation of $\cU$. By Corollary~\ref{c:root_gadget}, the root of $\cO$ subdivides an edge $e_\rho$ of the root gadget such that $e_\rho\ne\{r,v\}$. Let $\cO'$ be the orientation obtained from $\cU'$ by subdividing $e_\rho$ in $\cU'$ with a root vertex, directing the edge $\{v, p_n\}$ into $p_n$, and then directing each remaining edge as it is directed in $\cO$. Let $c_j=(x_i\vee x_{i'})$ be a $2$-clause. Then one of $x_i$ and $x_{i'}$, say $x_i$, is a tree vertex and $x_{i'}$ is a reticulation in $\cO$ and $\cO'$. It follows that each $x_i$ with $i\in\{1,2,\ldots,n\}$ is a tree vertex (resp. reticulation) in $\cO$ if and only if $x_i$ is a tree vertex (resp. reticulation) in $\cO'$, and the same applies to $x_{i'}$. Noting that the unique child $p_n$ of $r$ in $\cU$ is a tree vertex, it now follows that, as  $\cO$ is a funneled $\cC_A$-orientation of $\cU$, $\cO'$ is such an orientation of $\cU'$.

Second, suppose that $\cO'$ is a funneled $\cC_A$-orientation of $\cU'$. Let $e_\rho$ be the edge in $\cU'$ that contains the root in $\cO'$. By Corollary~\ref{c:root_gadget}, $e_\rho$ is an edge of the root gadget such that $e_\rho\ne\{v,p_n\}$. Hence, $e_\rho$ is also an edge of $\cU$. We now obtain an orientation $\cO$ for $\cU$ by subdividing $e_\rho$ with $\rho$, directing the two edges that are incident with $\rho$ away from $\rho$ and directing the remaining edges as follows. For each edge $e$ such that $e\ne e_\rho$ and $e$ is in $\cU$ and $\cU'$, $e$ has the same direction in $\cO$ as in  $\cO'$. Furthermore, for each leaf $\ell$ 
that is in $\cU$ but not in $\cU'$, the unique edge incident with $\ell$ is directed towards $\ell$. Since there exists a directed path from $\rho$ to each leaf in $\cO'$, it follows that $\cO'$ contains the arc $(v, p_n)$. For the same reason,  $(v, r)$  and, consequently, $(r, p_n)$ are arcs in $\cO$. Now, the only edges in $\cU$ that have not been assigned a direction  yet are edges of the connection gadgets whose non-leaf vertices lie on paths between $c_j$ and $x_{i'}$ that do not contain $x_i$ for each 2-clause $c_j = (x_i\vee x_{i'})$. By construction, $x_i$ is a tree vertex in $\cO'$ if and only if $x_{i'}$ is a reticulation in $\cO'$. Assume that $x_i$ is a tree vertex in $\cO'$. Then $(x_i, q)$ is an arc in $\cO'$ and $\cO$, where $q$ is the neighbour of $x_i$ in a connection gadget $\cG_2^1$. Furthermore, by Lemma~\ref{lem:connection_gadget}(iii), $(q',x_{i'})$ is an arc in $\cO'$, where $q'$ is the neighbour of $x_{i'}$ in $\cG_2^1$. In $\cU$, there exists a connection gadget  $\cG_2^2$ whose non-leaf vertices lie on paths from $c_j$ to $x_{i'}$  that do not contain $x_i$. Since $c_j$ is a degree-2 vertex in $\cU$ and $(q',c_j)$ is an arc of $\cO'$ (recall that $c_j = x_{i'}$ in $\cU'$), we direct the edges of $\cG_2^2$ as shown in Figure~\ref{fig:connection_gadget}(c) which yields a strongly $\mathcal{TC}$-compatible orientation for $\cG_2^2$. Then $c_j$ has in-degree 1 and out-degree 1 in $\cO$, and $c_j$ has a child in $\cG_2^2$ that is a tree vertex. Replacing Figure~\ref{fig:connection_gadget}(c) with Figure~\ref{fig:connection_gadget}(b), the case for when $x_i$ is a reticulation is similar and omitted. As $\cO'$ is a funneled $\cC_A$-orientation of $\cU'$, it now follows that we obtain a funneled $\cC_A$-orientation $\cO$ of $\cU$ by repeating the above steps for each $2$-clause.
\end{proof}

The next theorem is an immediate consequence of Theorem~\ref{thm:hardness_at_most_5} and Lemma~\ref{lem:help}.

\begin{theorem}\label{thm:hardness_3_or_5}
Let $\cC$ be the class of tree-child networks. Then \textsc{Funneled $\mathcal{C}_A$-Orientation ($3,5$)} is NP-complete if each non-leaf vertex has degree 3 or~5.   
\end{theorem}

\section{Hardness for Other Classes of Funneled Phylogenetic Networks}\label{sec:orientation-other-classes}
In this section, we establish NP-completeness for the following decision problem.
\begin{center}
\noindent\fbox{\parbox{.95\textwidth}{
\noindent\textsc{Funneled $\mathcal{C}_R$-Orientation ($5$)}\\
\textbf{Instance.} An unrooted phylogenetic network $\cU$ such that each non-leaf vertex has degree  5.\\
\textbf{Question.} Does there exist a funneled $\mathcal{C}_R$-orientation of $\cU$?
}}
\end{center}
In particular, we consider the four classes of tree-child, tree-sibling, normal, and reticulation-visible networks and both rooting variants $A$ and $B$. %
\begin{figure}
\includegraphics[width=.85\textwidth]{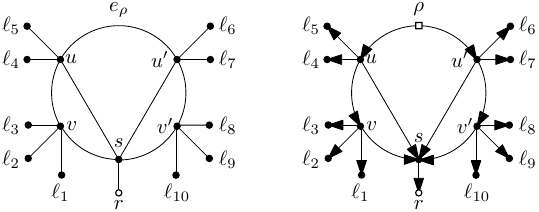}
\caption{Left: Root gadget with a single connector leaf $r$ and whose internal vertices all have degree 5. Right:  A funneled tree-child orientation for which the root is placed on the edge $e_\rho$.}
\label{fig:root_gadget_deg5}
\end{figure} 

\begin{figure}
\includegraphics[width=.375\textwidth]{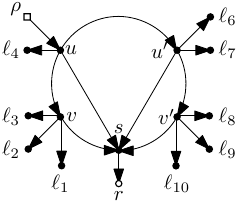}
\caption{A funneled tree-child orientation of the degree-5 root gadget that is shown on the left-hand side of Figure~\ref{fig:root_gadget_deg5} for which $\ell_5$ is chosen as root $\rho$. Note that the unique reticulation of this orientation is $s$.}
\label{fig:rooting_at_leaf}
\end{figure} 
 
We start by establishing properties of the {\it degree-$5$ root gadget} that is shown on the left-hand side of Figure~\ref{fig:root_gadget_deg5}. 

\begin{lemma}\label{lem:root_gadget_deg5}
Let $\cG$ be the degree-5 root gadget. Let $\cC$ be the class of tree-child networks, and let $\cC'$ the class of funneled phylogenetic networks. Then
\begin{enumerate}[(i)]
\item $\cG$  has a funneled $\mathcal{C}_A$-orientation such that $\rho$ subdivides $e_\rho = \{u,u'\}$, 
\item $\cG$  is $\cC'_A$-root-forcing, 
\item $\cG$  has a funneled $\mathcal{C}_B$-orientation such that $\rho=\ell_5$, and
\item $\cG$  is $\cC'_B$-root-forcing.
\end{enumerate}
\end{lemma}

\begin{proof}
The orientation that is shown in Figures~\ref{fig:root_gadget_deg5} (right-hand side) and~\ref{fig:rooting_at_leaf} is a funneled $\mathcal{C}_A$-orientation and a funneled $\mathcal{C}_B$-orientation, respectively, of the partner network of $\cG$ which only differs from $\cG$ by viewing $r$ as a non-connector leaf. This establishes (i) and (iii). Now, assume towards a contradiction that $\cG$ does not satisfy (ii). Then there exists a $\cC'_A$-orientation $\cO$ of $\cG$ such that $\rho$ subdivides $\{s,r\}$. Then  $(\rho, s)$ is an arc in $\cO$. Furthermore, each vertex in $\{u,u',v,v'\}$ has in-degree~1 and out-degree~4 since each such vertex has at least two adjacent leaves and $\cO$ is a $\cC'_A$-orientation. We next distinguish two cases.

First, suppose that $s$ has in-degree~1 and out-degree~4 in $\cO$. Then $(s, u)$ and $(s, v)$ are arcs in $\cO$. Since either $(u,v)$ or $(v,u)$ is an arc in $\cO$, $u$ or $v$ has in-degree at least $2$ and out-degree at least $2$; thereby contradicting that $\cO$ is a $\cC'_A$-orientation of $\cG$. Second, suppose that $s$ has in-degree~4 and out-degree~1. By symmetry, we may assume that $(u,u')$ is an arc in $\cO$. Since $u'$ and $v'$ have in-degree~1 and out-degree~4, it follows that $(u',v')$, $(u',s)$ and $(v',s)$ are arcs in $\cO$. Now, as $v$ has in-degree~1 and out-degree 4, we have that either $(s,v)$ and $(v,u)$, or $(u,v)$ and $(v,s)$ are arcs in $\cO$. In the former case, we get the directed cycle $(s,v), (v,u), (u,u'), (u',v'), (v',s)$. Hence, we may assume that $(u,v)$ and $(v,s)$ are arcs in $\cO$. As $s$ has in-degree 4 and out-degree 1, $(s,u)$ is an arc in $\cO$. This implies that $\cO$ contains the directed cycle $(u,v), (v, s), (s,u)$; thereby again contradicting that $\cO$ is a $\cC'_A$-orientation of $\cG$. Combining both cases establishes (ii). 

Lastly, assume towards a contradiction that $\cG$ does not satisfy (iv). Then there exists a $\cC'_B$-orientation $\cO$ of $\cG$ such that $r$ is chosen to be the root of $\cO$.  Then  $(r, s)$ is an arc in $\cO$. The proof of (iv) can now be established in the same way as the proof of (ii).
\end{proof}

To establish the next theorem, the construction from \textsc{Positive Not-All-Equal (2,3)-SAT}  as established in the last section cannot be used.  Instead, we reduce from \textsc{Positive 1-in-3 SAT}, replace the binary root  gadget as shown in Figure~\ref{fig:root_gadget} with a degree-5 root gadget, omit the connection gadget as shown in Figure~\ref{fig:connection_gadget} altogether,  and use clause vertices that each have two adjacent leaves and three adjacent variable vertices. Roughly speaking, the construction results in a phylogenetic network such that any funneled orientation of the constructed unrooted phylogenetic network enforces all clause vertices to be tree vertices by Observation~\ref{obs:rule1}.

\begin{theorem}\label{thm:hardness_exactly_5}
Let $\cC$ be the class of rooted phylogenetic networks. Then \textsc{Funneled $\mathcal{C}_A$-Orientation} $(5)$ is NP-complete.  
\end{theorem}
\begin{proof}
Let $\cU$ be an unrooted pseudo network, and let $\cO$ be an orientation of $\cU$ following Variant $A$. Since it can be checked in polynomial time if $\cO$ is a $\mathcal{C}_A$-orientation of $\cU$, it follows that \textsc{Funneled $\mathcal{C}_A$-Orientation}$(5)$ is in NP. 
We next establish NP-hardness using a polynomial-time reduction from the variant of \textsc{Positive 1-in-3 SAT}, where each variable appears exactly three times (see, e.g., Kratochv\'{i}l~\cite[Cor.\ 1]{kratochvil03} for a proof that the problem remains NP-complete under this restriction).

Let $\mathcal{I} = (V, C)$ be an instance of \textsc{Positive 1-in-3 SAT} such that each variable appears in exactly three clauses. Let $\cG_{n+1}^c$ be the caterpillar gadget with connector leaves $r,x_1, x_2,\ldots, x_n$, and let $\cG_1^r$ be the degree-5 root gadget with connector leaf~$r$.  We next construct an unrooted  phylogenetic network  $\cU$ in the following way.

\begin{enumerate}
\item Identify the connector leaf $r$ of $\cG_1^r$ with the connector leaf $r$ of $\cG_{n+1}^c$, and suppress the resulting degree-2 vertex.
\item Attach a new leaf $\ell_i$ to each connector leaf $x_i$ with $i\in\{1,2,\ldots,n\}$ via a new edge. 
\item Let $c_1,c_2,\ldots,c_m$ be new vertices, let $i\in\{1,2,\ldots,n\}$, and let $j\in\{1,2,\ldots,m\}$. For each $x_i$ and $c_j$ such that $x_i$ is a variable of $c_j$, add the edge $\{x_i,c_j\}$. Furthermore, attach two new leaves to each $c_j$.
\end{enumerate}
We may assume for the remainder of the proof that all leaves of $\cU$ have a distinct label and that the leaf set of $\cU$ is $X$. By construction, each non-leaf vertex in $\cU$ has degree 5.
Figure~\ref{fig:construction-deg5} shows and example of this construction for the following Boolean formula 
\begin{eqnarray}\label{eq:two}
(x_1\vee x_2\vee x_3)\wedge (x_2\vee x_3\vee x_4)\wedge  (x_3\vee x_4\vee x_5)\wedge \nonumber \\
 (x_4\vee x_5\vee x_6)\wedge  (x_5\vee x_6\vee x_1)\wedge  (x_6\vee x_1\vee x_2).
\end{eqnarray}

\begin{figure}
\includegraphics[width=.78\textwidth]{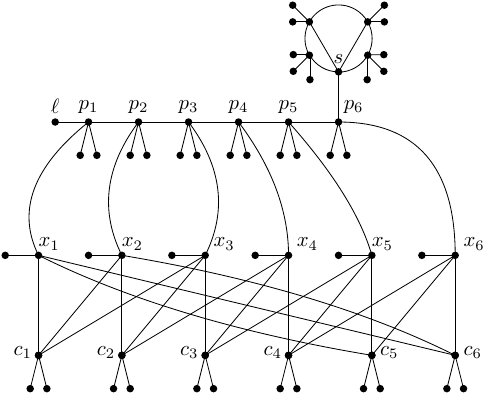}
\caption{The unrooted phylogenetic network that is constructed from the Boolean formula~(\ref{eq:two}) of  \textsc{Positive 1-in-3 SAT}. Details are given in the proof of Theorem~\ref{thm:hardness_exactly_5}. }
\label{fig:construction-deg5}
\end{figure} 

To complete the proof, we follow a similar approach as in the proof of Theorem~\ref{thm:hardness_at_most_5}.

\begin{sublemma}\label{claim-2}
$\cI$ is a yes-instance if and only if $\cU$ has a funneled $\mathcal{C}_A$-orientation.
\end{sublemma}
\begin{proof}
First, let $\beta\colon V \rightarrow \{T, F\}$ be a truth assignment that sets exactly one variable of each clause in~$C$ to $T$. We next construct a funneled $\cC_A$-orientation $\cO$ of $\cU$. As in the proof of Theorem~\ref{thm:hardness_at_most_5}, instead of referring to edges of $\cU$, we sometimes refer to edges of $\cG_1^r$ and $\cG_{n+1}^c$ when assigning a direction. We start by subdividing the edge $\{u,u'\}$ of $\cG_1^r$ with a new root vertex $\rho$ and direct the edges of the resulting root gadget as shown on the right-hand side of Figure~\ref{fig:root_gadget_deg5}. Noting that the edge $\{r,s\}$ in $\cG_1^r$ corresponds to the edge $\{p_n,s\}$ in $\cU$, this assignment of directions to edges implies that the arc $(s,p_n)$ is in $\cO$. We next direct the edges of $\cG_{n+1}^c$ (except for the edge $\{p_n,r\}$ which also corresponds to the edge $\{p_n,s\}$ in $\cU$ and has already been directed) as shown on the right-hand side of Figure~\ref{fig:caterpillar_gadget}. Now, for each $i\in\{1,2,\ldots,n\}$, let $c_j$, $c_{j'}$, and $c_{j''}$ be the clause vertices that are neighbours of $x_i$ in $\cU$. Depending on $\beta(x_i)$ we direct the three edges $\{c_{j},x_i\}$, $\{c_{j'},x_i\}$, and $\{c_{j''},x_i\}$ as follows.
\begin{enumerate}[(i)]
\item $(x_i,c_{j})$, $(x_i,c_{j'})$, and $(x_i,c_{j''})$ if $\beta(x_i)=T$, and
\item $(c_{j},x_i)$, $(c_{j'},x_i)$, and $(c_{j''},x_i)$ if $\beta(x_i)=F$.
\end{enumerate}
At this point, the only edges of $\cU$ that have not yet been assigned a direction are incident with a leaf and are directed towards that leaf in $\cO$.  By construction of $\cO$, the only vertex with in-degree 0 is $\rho$ and the set of vertices with out-degree 0 is $X$. Since $\beta$ sets exactly one variable of each clause to $T$, each clause vertex $c_j$ is a tree vertex in $\cO$ with  in-degree 1 and out-degree 4. Furthermore, each variable vertex $x_i$ is either a tree vertex with in-degree 1 and out-degree 4, or a reticulation with in-degree 4 and out-degree 1 in $\cO$. It now follows from Lemmas~\ref{lem:caterpillar_gadget} and~\ref{lem:root_gadget_deg5}(i) that $\cO$ is funneled. To see that $\cO$ is also acyclic, we can show that no vertex in $\{x_1,x_2,\ldots,x_n,c_1,c_2,\ldots,c_m\}$ lies on a directed cycle of $\cO$ by using the same `acyclicity argument' as that in the proof of~\ref{claim:hardness_at_most_5-sub}. Hence, $\cO$ is a funneled $\cC_A$-orientation of $\cU$.

Second, let $\cO$ be a funneled $\cC_A$-orientation of $\cU$. Using the degree-5 root gadget instead of the root gadget that is shown in Figure~\ref{fig:root_gadget}, and Lemma~\ref{lem:root_gadget_deg5}(ii) instead of Lemma~\ref{lem:root_gadget}(ii) in the proof of Corollary~\ref{c:root_gadget}, it is straightforward to show that $\rho$ subdivides an edge of $\cU$ that is an edge of $\cG_1^r$ and not $\{r,s\}$.  Hence, $(s, p_n)$ is an arc in $\cO$. It then follows by Lemma~\ref{lem:caterpillar_gadget}(ii)  that, for each $i \in \{1,2,\ldots, n\}$, $(p_i, x_i)$ is an arc in $\cO$. We now define a truth assignment $\beta \colon V \rightarrow \{T, F\}$ as follows. For each $i \in \{1,2,\ldots,n\}$, we set $\beta(x_i) = T$ if $x_i$ is a tree vertex and $\beta(x_i) = F$ if $x_i$ is a reticulation in $\cO$. Assume towards a contradiction that there exists a clause $c_j\in C$ with $c_j=(x\vee x'\vee x'')$ and $j\in\{1,2,\ldots,m\}$ such that $\beta$ does not set exactly one variable of $c_j$ to $T$. We distinguish three cases: 
\begin{enumerate}[(i)]
\item If $\beta(x) = \beta(x') = \beta(x'') = T$, then  $x, x'$ and $x''$ are tree vertices and, thus, $(x, c_j)$, $(x', c_j)$ and $(x'', c_j)$ are arcs in $\cO$. Since the other two vertices adjacent to $c_j$ are leaves, it follows that $c_j$ has in-degree 3 and out-degree 2; a contradiction to $\cO$ being a funneled orientation. 
\item If $\beta(x) = \beta(x') = T$ and $\beta(x'') = F$, then $x, x'$ are tree vertices and $x''$ is a reticulation and, thus, $(x, c_j)$, $(x', c_j)$ and $(c_j, x'')$ are arcs in $\cO$. But then, due to the two adjacent leaves of $c_j$, it follows that $c_j$ has in-degree 2 and out-degree 3; again a contradiction.
\item If $\beta(x) = \beta(x') = \beta(x'') = F$, then $x, x'$ and $x''$ are reticulations and, thus, $(c_j, x)$, $(c_j, x')$ and $(c_j, x'')$ are arcs in $\cO$. But then, due to the two adjacent leaves of $c_j$, it follows that $c_j$ has in-degree 0 and out-degree 5; a final contradiction.
\end{enumerate}
Hence, $\beta$ sets exactly one variable of each clause in $C$ to $T$.
\end{proof}
We conclude the proof of the theorem by noting that $\cU$ can be constructed in polynomial time and has a size that is polynomial in $n$ and $m$. 
\end{proof}

We now turn to \textsc{Funneled $\mathcal{C}_B$-Orientation} $(5)$, where $\cC$ is again the class of  rooted phylogenetic networks and the root is placed following Variant $B$. More precisely, given an unrooted phylogenetic network $\cU$, this variant chooses a vertex of $\cU$ to be the root and assigns a direction to each edge. Consider the degree-5 root gadget $\cG$ as shown on the left-hand side of Figure~\ref{fig:root_gadget_deg5}. Suppose that $\cU$ contains $\cG$ as a pending subgraph.  First, by Lemma~\ref{lem:root_gadget_deg5}(iv), any funneled $\cC_B$-orientation $\cO$ of $\cU$ chooses a vertex of $\cG$ that is not $r$ as the root. 
Second, Lemma~\ref{lem:root_gadget_deg5}(iii) shows that $\cG$ has a funneled $\cC_B$-orientation. The next corollary, which is a slight strengthening of a result by Bulteau et al.~\cite[Cor.\ 4]{bulteau23}, shows NP-hardness for graphs with maximum degree 5, now follows from an argument that is analogous to that in the proof of Theorem~\ref{thm:hardness_exactly_5}.

\begin{corollary}\label{c:hardness_general_exactly_5}
Let $\cC$ be the class of rooted phylogenetic networks. Then \textsc{Funneled $\mathcal{C}_B$-Orientation} $(5)$ is NP-complete. 
   
\end{corollary}

We remark that the root gadget as shown in Figure~\ref{fig:root_gadget} is not $\cC_R$-root-forcing for each $R\in\{A,B\}$ and can therefore not be used to establish Theorem~\ref{thm:hardness_exactly_5} or Corollary~\ref{c:hardness_general_exactly_5} for networks whose internal vertices have degree 3 or 5.

The following  lemma is the key ingredient in showing that \textsc{Funneled $\mathcal{C}_R$-Orientation} $(5)$ is NP-complete for the class of tree-child networks as well as for any class of rooted phylogenetic networks that contains the class of tree-child networks. 

\begin{lemma}\label{l:equiv-funneled}
Let $\cC$ be the class of rooted phylogenetic networks, let $\cC'$ be the class of tree-child networks. Furthermore, let $\cU$ be an unrooted phylogenetic network that is constructed from an instance of   \textsc{Positive 1-in-3 SAT} as described in the proof of Theorem~\ref{thm:hardness_exactly_5}. Then, for each $R\in\{A,B\}$, $\cU$ has a funneled $\cC_R$-orientation if and only if $\cU$ has a funneled $\cC'_R$-orientation.
\end{lemma}

\begin{proof}
We establish the proof for $R=A$. The proof for $R=B$ is similar and omitted. Clearly, any funneled $\cC'_A$-orientation of $\cU$ is also a funneled $\cC_A$-orientation of $\cU$. So suppose that $\cO$ is a funneled  $\cC_A$-orientation of $\cU$. The only vertex of $\cU$ that is not adjacent to a leaf is the vertex $s$ of the degree-5 root gadget. Since $(s,p_n)$ is an arc of $\cO$ as argued in the last paragraph of the proof of~\ref{claim-2}, it follows that $p_n$ is a child of $s$. Noting that $p_n$ has two adjacent leaves, it now follows that $p_n$ has in-degree 1 and out-degree 4 because $\cO$ is funneled. Thus $s$ has a child that is a tree vertex and $\cO$ is a funneled $\cC'_R$-orientation of $\cU$.
\end{proof}

\noindent Since the last lemma does not only hold for $\cC'$ being the class of tree-child networks but also for any class of rooted phylogenetic networks that contains the class of tree child networks, the next corollary is an immediate consequence of Theorem~\ref{thm:hardness_exactly_5}, Corollary~\ref{c:hardness_general_exactly_5}, and Lemma~\ref{l:equiv-funneled}.

\begin{corollary}
Let $\cC$ be the class of tree-child, tree-sibling, or reticulation-visible networks. Then \textsc{Funneled $\mathcal{C}_R$-Orientation} $(5)$ is NP-complete for each $R\in\{A,B\}$
\end{corollary}

Lastly, we establish NP-completeness of \textsc{Funneled $\mathcal{C}_R$-Orientation} $(5)$ for the class of normal networks. Let $e$ be an edge of an unrooted phylogenetic network $\cU$ on $X$. We refer to the operation of subdividing $e$ with a new vertex $v$ and adding the three edges $\{v,\ell\}$, $\{v,\ell'\}$, and $\{v,\ell''\}$ such that $\ell$, $\ell'$, and $\ell''$ are leaves that are not contained in $X$ as  {\it attaching three leaves to $e$.} Now, let $\cU$ be an unrooted phylogenetic network that is constructed from an instance $\cI$ of  \textsc{Positive 1-in-3 SAT} as described in the proof of Theorem~\ref{thm:hardness_exactly_5}. Obtain an unrooted phylogenetic network $\cU'$ from $\cU$ by attaching three leaves to each edge $e$ that is incident with a vertex in $\{x_1,x_2,\ldots,x_n\}$ and not incident with a leaf as well as to the edges $\{s,u\}$, $\{s,u'\}$, and $\{s,p_n\}$, where $s$, $u$, and $u'$ are the vertices as shown on the left-hand side of Figure~\ref{fig:root_gadget_deg5}. We refer to $(\cU,\cU')$ as the {\it  network pair associated with $\cI$}. 

\begin{lemma}\label{l:normal}
Let $\cC$ be the class of rooted phylogenetic networks, and let $\cC'$ be the class of normal networks. Furthermore, let $(\cU,\cU')$ be the  network pair associated with an instance of  \textsc{Positive 1-in-3 SAT}. Then, for each $R\in\{A,B\}$, $\cU$ has a funneled $\cC_R$-orientation if and only if $\cU'$ has a funneled $\cC'_R$-orientation.
\end{lemma}

\begin{proof}
Let $E$ be the subset of edges of $\cU$ that precisely contains each edge that is subdivided in obtaining $\cU'$ from $\cU$ by attaching three leaves.  Clearly, each edge that is in $\cU$ and not in $E$ is also an edge of $\cU'$. Moreover, for each edge $e=\{u,w\}$ in $E$, there exist edges $\{u,v\},\{v,w\},\{v,\ell\},\{v,\ell'\}$, and $\{v,\ell''\}$ in $\cU'$, where $\ell$ $\ell'$, and $\ell'$ are leaves. We refer to those five edges as {\it five edges associated with $e$ in $\cU'$}. Let $\cO$ be a $\cC_R$-orientation of $\cU$. If $R=A$, then it follows from the last paragraph of the proof of~\ref{claim-2} that $(s,p_n)$ is an arc of $\cO$. Similarly, if $R=B$, then follows from the paragraph prior to Corollary~\ref{c:hardness_general_exactly_5} that $(s,p_n)$ is again an arc in  $\cO$. We freely use the existence of the arc $(s,p_n)$ throughout the remainder of the proof.

First, suppose that $\cU$ has a funneled $\cC_A$-orientation $\cO$. Let $e_\rho$ be the edge of $\cU$ that is subdivided in obtaining $\cO$ from $\cU$. Apply the following three steps to obtain an orientation $\cO'$ of $\cU'$.
\begin{enumerate}[(1)]
\item \label{step:root-edge} If $e_\rho\notin E$, then subdivide $e_\rho$ of $\cU'$ with a vertex $\rho$ and direct the two edges incident with $\rho$ so that $\rho$ has in-degree 0. On the other hand, if $e_\rho=\{u_e,w_e\}$ is an edge in $E$, recall that $\{u_e,v_e\},\{v_e,w_e\}$ is a path of length two in $\cU'$ such that $v_e$ is adjacent to three leaves $\ell_e$ $\ell'_e$, and $\ell'_e$. Then subdivide $\{u_e,v_e\}$ with a vertex $\rho$ and let $(\rho,u_e),(\rho,v_e),(v_e,w_e),(v_e,\ell_e),(v_e,\ell_e'),$ and $(v_e,\ell_e'')$ be arcs in~$\cO'$.
\item \label{step:same-edge} Direct each edge of $\cU'$ that is also an edge in $\cU$ such that it has the same direction as in $\cO$.
\item \label{step:remaining-edge} For each edge $e=\{u,w\}$ in $E\setminus \{e_\rho\}$, direct the five edges associated with $e$ in $\cU'$ in the following way. If $(u,w)$ is an arc in $\cO$, then $(u,v),(v,w),(v,\ell),(v,\ell')$, and $(v,\ell'')$ are arcs in $\cO'$ and, otherwise,  $(w,v),(v,u),(v,\ell),(v,\ell')$, and $(v,\ell'')$ are arcs in $\cO'$.
\end{enumerate}
Since $\cO$ is a funneled $\cC_A$-orientation of $\cU$, it follows that $\cO'$ is a funneled $\cC_A$-orientation of $\cU'$. Moreover, as  $(s,p_n)$ is an arc of $\cO$, there is a directed path from $s$ to $p_n$ of length two in $\cO'$ whose middle vertex is a tree vertex. Since $\cO'$ is funneled and, each vertex except for $s$ is adjacent to a leaf, it  follows that $\cO'$ is a funneled $\cC_A$-orientation for $\cU'$ that is also tree-child. To see that $\cO'$ is in fact a funneled $\cC'_A$-orientation of $\cU'$, it remains to argue that $\cO'$ has no shortcut. Let $v$ be a reticulation of $\cO'$, and let $p$ and $p'$ be two distinct parents of $v$. Since $v\in \{s,x_1,x_2,\ldots,x_n\}$ it is straightforward to check that $p$ and $p'$ are two vertices of $\cO'$ each with three adjacent leaves. Assume that there is a directed path from $p$ to $p'$ in $\cO'$. Then there exists an arc $(p,u)$ such that $u$ is not a leaf and $u\ne v$. This implies that $p$ has in-degree 0 and out-degree 4; a contradiction. It now follows that $(p,v)$ is not a shortcut and, by a symmetric argument, $(p',v)$ is not a shortcut either. Thus, $\cO'$ is a funneled $\cC'_A$-orientation of $\cU'$.

Second, suppose that $\cU$ has a funneled $\cC_B$-orientation $\cO$. The root $\rho$ of $\cO$ is a vertex of $\cU$ and $\cU'$. Obtain an orientation $\cO'$ of $\cU'$  by following Steps~\eqref{step:same-edge} and~\eqref{step:remaining-edge} so that Step~\eqref{step:remaining-edge} is applied to each edge in $E$. Using the same argument as in the previous paragraph, it follows that $\cO'$ is a funneled $\cC'_A$-orientation of $\cU'$. 

Third, suppose that $\cU'$ has a funneled  $\cC'_A$-orientation $\cO'$.  Let $e_\rho$ be the edge of $\cU'$ that is subdivided in obtaining $\cO'$ from $\cU'$. Obtain an orientation $\cO$ of $\cU$ by reversing Steps (1)--(3) from above as follows.
\begin{enumerate}[(1')]
\item \label{step:root-edge-rev} If $e_\rho$ is an edge of $\cU$, subdivide $e_\rho$ in $\cU$ with a vertex $\rho$ and direct the two edges incident with $\rho$ such that $\rho$ has in-degree 0. On the other hand, if $e_\rho$ subdivides one of the five edges associated with an edge $e\in E$, then subdivide $e$ in $\cU$ with a vertex $\rho$ and direct the two edges incident with $\rho$ such that $\rho$ has in-degree 0. 
\item \label{step:same-edge-rev} Direct each edge of $\cU$ that is also an edge in $\cU'$ such that it has the same direction as in $\cO'$.
\item \label{step:remaining-edge-rev} For each $e=\{u,w\}$  in $E\setminus\{e_\rho\}$, consider the five edges associated with $e$ in $\cU'$. If  $(u,v),(v,w),(v,\ell),(v,\ell')$, and $(v,\ell'')$ are arcs in $\cO'$, then $(u,w)$ is an arc in $\cO$ and, if $(w,v),(v,u),(v,\ell),(v,\ell')$, and $(v,\ell'')$ are arcs in $\cO'$, then $(w,u)$ is an arc in $\cO$. 
\end{enumerate}
Since $\cO'$ is funneled, one of the two cases described in Step~(\ref{step:remaining-edge-rev}') applies. Furthermore, since $\cO'$ is a funneled $\cC'_A$-orientation of $\cU'$, it follows by construction that $\cO$ is a funneled $\cC_A$-orientation of $\cU$.

Lastly, suppose that  $\cU'$ has a funneled  $\cC'_B$-orientation $\cO'$. Let $\rho$ be the root of $\cO'$. We first show that $\rho$ is a vertex of $\cU$. Assume towards a contradiction that $\rho$ is not a vertex of $\cU$. Let $\{s,s'\}$ be the unique edge of $\cU'$, where $s$ is as shown on the left hand-side of Figure~\ref{fig:root_gadget_deg5} and $s'$ is the unique internal vertex introduced by attaching three leaves to $\{r,s\}$. By arguments similar to the ones used to establish Lemma~\ref{lem:root_gadget_deg5}(ii), it follows that $(s,s')$ is an arc in $\cO'$. Hence $\rho$ is a vertex that is introduced in attaching three leaves to $\{u,s\}$ or $\{u',s\}$. By symmetry, we may assume that $\rho$ is one of the four vertices $t,\ell,\ell'$ and $\ell''$, where $t$ is the unique internal vertex introduced by attaching three leaves to $\{u,s\}$. Regardless of which vertex is chosen as $\rho$, $(t,u)$ and $(t,s)$ are arcs in $\cO'$ because $\cO'$ is funneled. Furthermore, as $u$ and $v$ have in-degree 1 and out-degree~4, it follows that $(u,v)$ and, in turn, $(v,s)$ are arcs in $\cO'$. Thus, $(t,s)$ is a shortcut and therefore $\cO'$ is not a $\cC'_B$-orientation of $\cU'$; a contradiction. Hence, $\rho$ is a vertex of $\cU$, and we obtain an orientation $\cO$ of $\cU$ by following Steps~(\ref{step:same-edge-rev}') and~(\ref{step:remaining-edge-rev}') such that Step~(\ref{step:remaining-edge-rev}') is applied to each edge in $E$. Then it follows again by construction that $\cO$ is a funneled $\cC_B$-orientation of $\cU$. 
\end{proof}

The next corollary follows from Theorem~\ref{thm:hardness_exactly_5}, Corollary~\ref{c:hardness_general_exactly_5}, and Lemma~\ref{l:normal}.
\begin{corollary}
Let $\cC$ be the class of normal networks. Then \textsc{Funneled $\mathcal{C}_R$-Orientation} $(5)$ is NP-complete for each $R\in\{A,B\}$.
\end{corollary}

\bibliographystyle{alpha}

\begin{thebibliography}{99}

\bibitem{bulteau23}
Bulteau, L., Weller, M., Zhang, L. (2023). On turning a graph into a phylogenetic network, hal-04085424.

\bibitem{cardona09}
Cardona, G., Rossell\'{o}, F., Valiente, G. (2009). Comparison of tree-child phylogenetic networks, IEEE/ACM Transactions on Computational Biology and Bioinformatics, 6, pp. 552--569.

\bibitem{cardona08}
Cardona, G., Llabr\'{e}s, M., Rossell\'{o}, F., Valiente, G. (2008). A distance metric for a class of tree-sibling phylogenetic networks, Bioinformatics, 24:1481--1488.

\bibitem{dehghan15}
Dehghan, A., Sadeghi, M., Ahadi, A. (2015). On the complexity of deciding whether the regular number is at most two. Graphs and Combinatorics, 31:1359--1365.

\bibitem{garvardt23}
Garvardt, J., Renken, M., Schestag, J., Weller, M. (2023). Finding degree-constrained acyclic orientations. In: 18th International Symposium on Parameterized and Exact Computation, pp. 19:1–19:14

\bibitem{huber22}
Huber, K. T., van Iersel, L., Janssen, R., Jones, M., Moulton, V., Murakami, Y., Semple, C. (2024). Orienting undirected phylogenetic networks. Journal of Computer and System Sciences 140:103480.

\bibitem{iersel10}
van Iersel, L., Semple, C, Steel, M. (2010) Locating a tree in a phylogenetic network. Information Processing Letters 110:1037--1043.

\bibitem{janssen18}
Janssen, R., Jones, M. and Erd{\H{o}}s, P. L., van Iersel, L., Scornavacca, C. (2018). Exploring the tiers of rooted phylogenetic network space using tail moves. Bulletin of Mathematical Biology, 80:2177--2208.

\bibitem{kong22}
Kong, S., Pons, J. C., Kubatko, L.,  Wicke, K. (2022). Classes of explicit phylogenetic networks and their biological and mathematical significance. Journal of Mathematical Biology, 84:47.

\bibitem{kratochvil03}
Kratochv\'{i}l, J. (2003). Complexity of hypergraph coloring and Seidel’s switching. In: WG 2003, Lecture Notes in Computer Science 2880, pp. 297--308.

\bibitem{maeda23}
Maeda, S., Kaneko, Y., Muramatsu, H., Murakami, Y., Hayamizu, M. (2023). Orienting undirected phylogenetic networks to tree-child networks, arXiv:2305.10162.

\bibitem{willson10}
Willson, S. J. (2010). Properties of normal phylogenetic networks. Bulletin of
Mathematical Biology, 72:340--358.

\end{thebibliography}

\end{document}